\documentclass{article}

\usepackage[preprint]{neurips_2026}

\usepackage[utf8]{inputenc}
\usepackage[T1]{fontenc}
\usepackage{lmodern}
\usepackage{hyperref}
\hypersetup{hidelinks}
\usepackage{url}
\usepackage{booktabs}
\usepackage{amsfonts}
\usepackage{amsmath}
\usepackage{amssymb}
\usepackage{amsthm}
\usepackage{nicefrac}
\usepackage{microtype}
\usepackage{graphicx}
\usepackage{xcolor}
\usepackage{algorithm}
\usepackage{algorithmic}
\usepackage{multirow}
\usepackage{subcaption}

\newtheorem{theorem}{Theorem}
\newtheorem{assumption}{Assumption}
\newtheorem{corollary}{Corollary}

\title{LLaDA-TTS: Unifying\\Speech Synthesis and Zero-Shot Editing\\via Masked Diffusion Modeling}

\author{%
  Xiaoyu Fan\thanks{Work done during an internship at Bairong.} \quad
  Huizhi Xie\thanks{Corresponding author.} \quad
  Wei Zou \quad
  Yunzhang Chen \\[0.5em]
  BRVoice Team, Bairong, Inc., China \\
  \texttt{\{xiaoyu.fan, huizhi.xie, wei.zou, yunzhang.chen\}@brgroup.com}
}

\begin{document}

\maketitle

\begin{abstract}
Large language model (LLM)-based text-to-speech (TTS) systems achieve remarkable naturalness via autoregressive (AR) decoding, but require $N$ sequential steps to generate $N$ speech tokens. We present \textbf{LLaDA-TTS}, which replaces the AR LLM with a masked diffusion model that completes generation in a fixed number of parallel steps, decoupling inference latency from sequence length. Remarkably, using only 50 hours of fine-tuning data, we successfully transfer a pretrained AR checkpoint to the masked diffusion paradigm via bidirectional attention. At 64 steps, LLaDA-TTS achieves 0.98\% CER (zh) and 1.96\% WER (en) on Seed-TTS-Eval, matching the original CosyVoice 3 baseline performance while delivering a $2\times$ LLM-stage speedup---a notable acceleration achieved despite the absence of KV cache, an optimization the AR baseline heavily relies on. Beyond acceleration, the bidirectional architecture naturally enables \emph{zero-shot speech editing}---including word-level insertion, deletion, and substitution---without any additional training. Theoretically, we prove that AR-pretrained weights are near-optimal for bidirectional masked prediction under the locality property of acoustic tokens, explaining this rapid convergence. This general method modifies only the attention mask and objective, applying seamlessly to any LLM-based AR TTS system. Code and audio samples will be available at \url{https://deft-piroshki-b652b5.netlify.app/}.
\end{abstract}

\section{Introduction}
\label{sec:intro}

The convergence of large language models (LLMs) and neural speech synthesis has produced TTS systems that rival human speech in naturalness and speaker fidelity~\citep{cosyvoice3,seedtts,valle}. These systems follow a pipeline where an LLM autoregressively generates discrete speech tokens, and a second-stage model converts them to waveforms. While highly effective, the AR bottleneck remains fundamental: generating a 10-second utterance at 25\,tokens/s requires ${\sim}250$ sequential forward passes, making latency scale linearly with output length.

\textbf{Masked discrete diffusion models}~\citep{llada,dream,mdlm} offer a principled alternative. Starting from a fully masked sequence, these models iteratively predict and reveal tokens in parallel, decoupling inference cost from output length. LLaDA~\citep{llada} established the $1/t$-weighted cross-entropy objective as a valid ELBO, and Dream~\citep{dream} showed that initializing from AR weights via \emph{label shift} enables efficient convergence.

We propose \textbf{LLaDA-TTS}, applying masked discrete diffusion to the speech token generation stage of LLM-based TTS. Our approach replaces only the LLM's causal attention mask and training objective---the tokenizer, prompt format, and acoustic model are untouched. This simplicity means the method applies to any LLM-based AR TTS system using discrete speech tokens (e.g., VALL-E~\citep{valle}, Seed-TTS~\citep{seedtts}, Spark-TTS~\citep{sparktts}), not just the CosyVoice\,3~\citep{cosyvoice3} backbone used in our experiments.

Beyond acceleration, we discover two emergent phenomena in LLaDA-TTS. First, despite fully bidirectional attention, the unmasking order exhibits a \emph{predominantly left-to-right} progression, indicating that sequential priors from AR pretraining persist in the diffusion regime---yet with confidence-based deviations that exploit bidirectional context for improved coherence. Second, specific attention heads in LLaDA-TTS's Transformer develop precise \emph{monotonic alignment} between text and speech positions---an attention pattern impossible under causal masking---enabling zero-shot speech editing via selective masking and regeneration, without additional training. We ground these findings theoretically by proving that, under a physically motivated $\varepsilon$-forward dependence of acoustic tokens, AR-pretrained predictions are provably near-optimal for bidirectional masked prediction (Theorem~\ref{thm:kl_bound}), explaining both the effectiveness of initialization and the emergent left-to-right dynamics.

Our contributions:
\begin{enumerate}
    \item We adapt a pretrained AR TTS model into a masked diffusion decoder. The method is architecture-agnostic and applicable to any LLM-based TTS system.

    \item On Seed-TTS-Eval, LLaDA-TTS achieves 0.98\% CER (zh) and 1.96\% WER (en) at 64 steps, matching the accuracy of the CosyVoice 3 baseline while delivering a $2\times$ computational speedup---an efficiency gained despite lacking the KV cache advantage inherent to AR models.

    \item We discover that specific attention heads in LLaDA-TTS precisely capture text-to-speech alignment, and that AR weight initialization naturally induces an approximate left-to-right decoding progression during the unmasking process.

    \item Empowered by the bidirectional attention mechanism, LLaDA-TTS enables a zero-shot speech editing pipeline---performing word-level insertion, deletion, and substitution---without requiring any additional training.
\end{enumerate}

\section{Related Work}
\label{sec:related}

\paragraph{LLM-based TTS.}
VALL-E~\citep{valle} pioneered generating discrete speech tokens with LLMs for zero-shot TTS. Seed-TTS~\citep{seedtts} scaled this to industrial quality; the CosyVoice series~\citep{cosyvoice1,cosyvoice,cosyvoice3} progressively improved via supervised tokenization, streaming, and scaled LLMs; F5-TTS~\citep{f5tts} used continuous flow matching; and Spark-TTS~\citep{sparktts} proposed single-stream decoupled tokens. All inherit the sequential AR bottleneck.

\paragraph{Discrete diffusion for language.}
D3PM~\citep{austin2021d3pm} established the theoretical framework for discrete diffusion. MDLM~\citep{mdlm} showed the absorbing-state (masked) case yields strong perplexity. LLaDA~\citep{llada} provided the $1/t$-weighted ELBO and matched AR models at 8B scale. Dream~\citep{dream} introduced AR initialization via label shift and matched LLaDA with 5$\times$ less data, demonstrating strong planning capabilities. These advances motivate our application to TTS.

\paragraph{NAR speech generation and editing.}
MaskGCT~\citep{maskgct} applied masked generation to TTS in a two-stage pipeline with MaskGIT-style~\citep{maskgit} decoding. Our work differs in using the ELBO-derived $1/t$ objective, single-stage generation, and AR initialization. For speech editing, VoiceCraft~\citep{voicecraft} requires specialized training, while LLaDA-TTS performs editing as a zero-cost byproduct of its bidirectional architecture.

\paragraph{Theoretical analysis of AR-diffusion connections.}
The relationship between autoregressive and diffusion models has been explored in continuous~\citep{austin2021d3pm} and discrete~\citep{llada} settings. The ELBO framework~\citep{mdlm,llada} provides a principled training objective, and Dream~\citep{dream} showed empirically that AR initialization is effective. However, a theoretical explanation for \emph{why} AR initialization works particularly well for specific modalities has been lacking. Our analysis fills this gap for speech: we prove that the temporal locality of acoustic tokens makes AR predictions near-optimal for bidirectional masked prediction, providing the first formal justification for AR-to-diffusion transfer in TTS.

\section{Method}
\label{sec:method}

We retain the standard LLM-based TTS pipeline---supervised speech tokenization, LLM-based token generation, and flow matching vocoder---but replace the causal AR LLM with a bidirectional masked diffusion model. Figure~\ref{fig:architecture} illustrates the architecture.

\begin{figure}[t]
    \centering
    \includegraphics[width=0.95\textwidth]{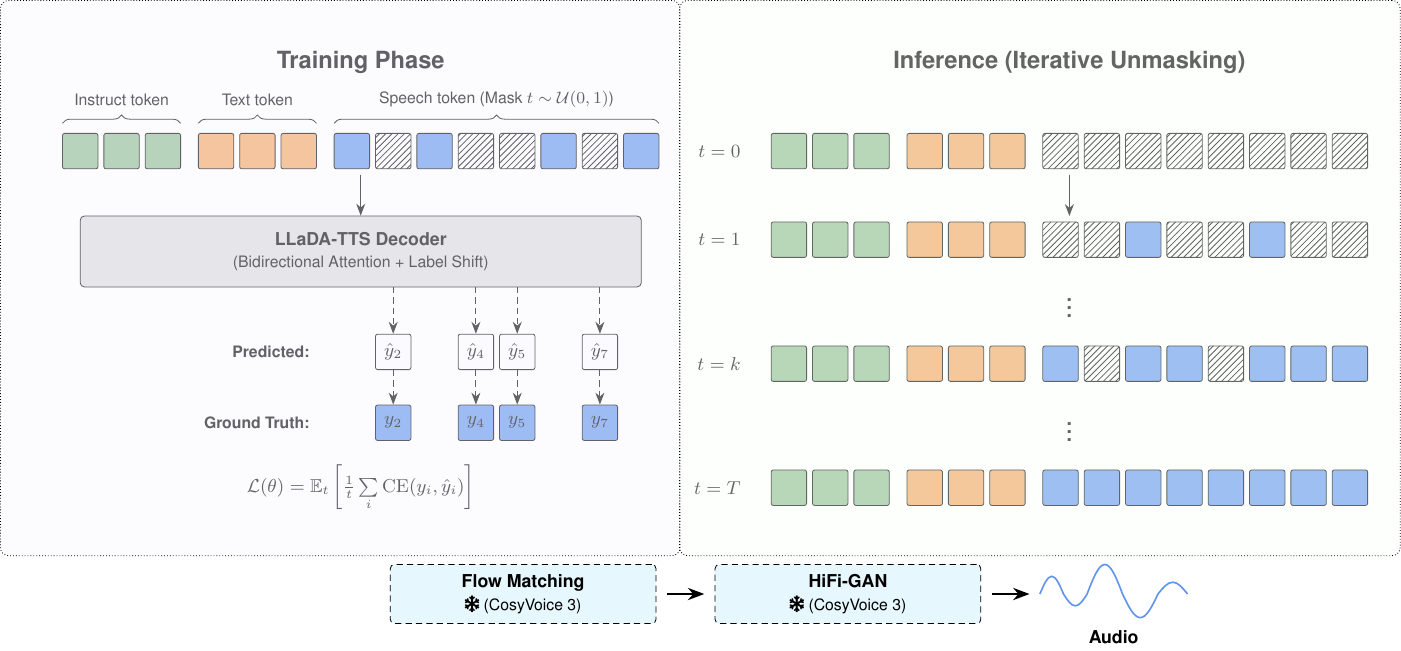}
    \caption{LLaDA-TTS architecture overview. A bidirectional Transformer (Qwen2) iteratively unmasks speech tokens in $T$ steps. The text encoder, sequence format, and downstream flow matching vocoder remain identical to the AR baseline.}
    \label{fig:architecture}
\end{figure}

\subsection{Architecture Adaptation}

\paragraph{Bidirectional attention.} We replace the causal mask with full bidirectional attention, allowing every token to attend to all others (only padding is masked). This enables each masked position to condition on both preceding and following context.

\paragraph{Sequence format.} The input sequence retains the standard structure:
\([\texttt{SOS},\allowbreak\, \mathbf{t}_\text{instruct},\allowbreak\, \mathbf{t}_\text{text},\allowbreak\, \texttt{SEP},\allowbreak\, \mathbf{s}_\text{prompt},\allowbreak\, \widetilde{\mathbf{s}},\allowbreak\, \texttt{EOS}]\),
where~$\widetilde{\mathbf{s}}$ denotes partially masked target speech tokens. Only target tokens participate in masking; text, prompt, and special tokens are always visible.

\paragraph{Label shift.} Following Dream~\citep{dream}, we preserve the AR convention where hidden state $i$ predicts token $i{+}1$, implemented by shifting logits: $\hat{\mathbf{y}}_i = \text{logits}_{i-1}$. This enables direct initialization from AR weights without retraining the output projection---the critical bridge for knowledge transfer.

\subsection{Training}

\paragraph{Initialization.} All parameters are loaded from a pretrained AR checkpoint (Qwen2-0.5B backbone), including the Transformer, speech embedding, and decoder. Only the new \texttt{[MASK]} token embedding is randomly initialized. This transfers linguistic and acoustic knowledge from AR pretraining, enabling fast convergence~\citep{dream}.

\paragraph{Masked diffusion objective.} For each sample, we sample $t \sim \mathcal{U}(0,1)$, independently mask each target speech token with probability $t$, and compute the $1/t$-weighted cross-entropy loss only on masked positions:
\begin{equation}
    \mathcal{L}(\theta) = \mathbb{E}_{t} \left[ \frac{1}{t} \sum_{i:\, s_t^i = \texttt{[MASK]}} -\log p_\theta(s_0^i \mid \mathbf{s}_t) \right]
    \label{eq:loss}
\end{equation}
The $1/t$ weighting upweights predictions near completion where each remaining token carries more information, corresponding to the ELBO derivation in LLaDA~\citep{llada}.

\paragraph{Training data.} We train on 6{,}000\,h of the Emilia dataset~\citep{emilia} (58\% Chinese, 42\% English; 2.5\,M utterances). Training uses 7$\times$A100 GPUs with DDP, Adam ($\text{lr}=10^{-5}$), gradient clipping at 5.0, and FP16 mixed precision.

\subsection{Why AR Initialization Works: A Theoretical Perspective}
\label{sec:theory}

The strong empirical performance of AR initialization (Table~\ref{tab:ablation}) admits a theoretical explanation rooted in the temporal structure of acoustic tokens.

\begin{assumption}[$\varepsilon$-Forward Dependence]
\label{ass:eps_forward}
A discrete acoustic token sequence $X = (x_1, \ldots, x_N)$ at frame rate $f$ satisfies $\varepsilon$-forward dependence if for all positions $i$ and any subset $S \subseteq \{1, \ldots, i{-}1\}$ of observed past positions:
\begin{equation}
I(x_i;\, X_{>i} \mid X_S) \leq \varepsilon
\label{eq:eps_forward}
\end{equation}
\end{assumption}

\noindent At 25\,Hz (40\,ms/frame), speech tokens satisfy this with small $\varepsilon$: articulatory parameters (tongue position, lip aperture, glottal state) evolve on timescales of 50--100\,ms~\citep{coverinfo}, so consecutive frames are highly redundant and the causal direction dominates information flow. The condition requires the bound for \emph{any} partial past observation $X_S$, which is physically natural---the information that future tokens carry about the current frame is inherently limited by the causal dynamics of speech production, regardless of which past frames are observed.

\begin{theorem}[Bounded Suboptimality of AR Initialization]
\label{thm:kl_bound}
Under Assumption~\ref{ass:eps_forward}, let $\tilde{X}$ be a partial masking of $X$ where each token is independently revealed with probability $1{-}\tau$. For any masked position $i$, the expected KL divergence between the optimal bidirectional predictor and the left-context-only (AR) predictor satisfies:
\begin{equation}
\mathbb{E}_{\tilde{X}}\!\left[ D_{\mathrm{KL}}\!\left( P(x_i \mid \tilde{X}_{<i}, \tilde{X}_{>i}) \;\big\|\; P(x_i \mid \tilde{X}_{<i}) \right) \right] \leq \varepsilon
\label{eq:kl_bound}
\end{equation}
\end{theorem}

\noindent\textit{Proof sketch.} The expected KL equals the conditional mutual information $I(x_i;\, \tilde{X}_{>i} \mid \tilde{X}_{<i})$. Conditioning on the independent masking pattern $M$ and decomposing, the term involving $M$ vanishes by $M \perp X$. For each fixed mask, the data processing inequality gives $I(x_i;\, X_{R_{>i}} \mid X_{R_{<i}}) \leq I(x_i;\, X_{>i} \mid X_{R_{<i}}) \leq \varepsilon$, where $R_{<i}, R_{>i}$ are the revealed positions. Averaging over masks preserves the bound. The full derivation is provided in the supplementary material. \hfill$\square$

\smallskip

Theorem~\ref{thm:kl_bound} implies that the optimal AR predictor already captures all but $\varepsilon$ bits of the per-token predictive distribution at \emph{every} diffusion step. This explains the dramatic convergence advantage of AR initialization: at matched compute, AR-initialized models achieve 0.98\% CER versus 45.27\% from scratch (Table~\ref{tab:ablation}), a gap directly accounted for by the $\varepsilon$-bounded suboptimality of AR predictors in the bidirectional regime.

\begin{corollary}[Emergence of Left-to-Right Unmasking]
\label{cor:wavefront}
During iterative unmasking, positions with smaller index $i$ have access to the always-visible text and prompt prefix, which provides strong left conditioning. By Theorem~\ref{thm:kl_bound}, the left-only prediction is near-optimal for these positions, yielding high confidence scores---and thus early unmasking. This cascading effect produces the left-to-right wavefront observed in Figure~\ref{fig:unmasking}.
\end{corollary}

\paragraph{Implications for training.} Theorem~\ref{thm:kl_bound} also suggests that, at high masking rates ($\tau \to 1$), the masked diffusion model's task closely resembles AR prediction---few tokens are revealed, and the left-context advantage dominates. This explains why AR-initialized models converge quickly: the initial weights are already near-optimal for the high-masking regime and only need adaptation for low-masking-rate predictions where more tokens are revealed and bidirectional context becomes useful. The $1/t$-weighting in the training objective (Eq.~\ref{eq:loss}) upweights precisely these low-masking-rate steps, focusing learning capacity where AR initialization provides the least advantage.

\subsection{Inference: Iterative Unmasking}

Generation proceeds from fully masked to fully revealed in $T$ steps (Algorithm~\ref{alg:inference}). At each step: (1) predict logits for all masked positions via one forward pass, (2) sample candidate tokens with temperature-scaled nucleus sampling, (3) compute confidence scores, and (4) unmask the top-$K_k$ most confident predictions, where $K_k$ follows a linear schedule. The total cost is $T$ forward passes, independent of sequence length.

\begin{algorithm}[t]
\caption{LLaDA-TTS Inference}
\label{alg:inference}
\begin{algorithmic}[1]
\REQUIRE Text $\mathbf{t}$, prompt $\mathbf{s}_\text{prompt}$, steps $T$, temperature $\tau$, confidence $\mathcal{C}$
\STATE $\mathbf{s} \leftarrow [\texttt{MASK}]^N$; \quad $\{t_k\} \leftarrow \text{linspace}(1, \epsilon, T{+}1)$
\FOR{$k = 0$ \textbf{to} $T{-}1$}
    \STATE $\hat{\mathbf{y}} \leftarrow \text{ShiftLogits}(\text{LLM}([\texttt{SOS}; \mathbf{t}; \texttt{SEP}; \mathbf{s}_\text{prompt}; \mathbf{s}; \texttt{EOS}]))$
    \FOR{each masked $i$}
        \STATE $\hat{s}_i \sim \text{TopP}(\hat{\mathbf{y}}_i / \tau)$; \quad $c_i \leftarrow \mathcal{C}(\hat{\mathbf{y}}_i)$
    \ENDFOR
    \STATE Unmask top-$\lfloor |\mathcal{M}| \cdot (1 - t_{k+1}/t_k) \rfloor$ positions by confidence
\ENDFOR
\end{algorithmic}
\end{algorithm}

We use top-$k$ margin confidence ($c_i = p_1(i) - p_2(i)$), temperature $\tau{=}0.986$, top-$p{=}0.586$, and confidence temperature $\tau_c{=}0.424$ (selected via 300 Optuna trials).

\subsection{Speech Editing via Selective Masking}
\label{sec:method_editing}

The bidirectional architecture naturally enables speech editing. Given existing speech tokens and a text edit, we: (1)~align text to speech tokens via LLaDA-TTS's attention at a specific head (L16-H2 or L11-H5, identified by systematic evaluation over all layer-head pairs; see Section~\ref{sec:exp_emergent}), (2)~mask the affected region with context margins ($C{=}5$ tokens for substitution, $S{=}3$ for insertion/deletion) for smooth transitions, and (3)~regenerate via iterative unmasking with surrounding tokens frozen. No additional training is required. Figure~\ref{fig:speech_edit} illustrates the full pipeline for an example substitution edit.

Formally, this computes the posterior $P(\mathbf{s}_{\text{edit}} \mid \mathbf{s}_{\text{prefix}}, \mathbf{s}_{\text{suffix}}, \mathbf{t}_{\text{new}})$, conditioning on both prefix and suffix context. Unlike AR models that can only condition on the left, the bidirectional attention of masked diffusion naturally marginalizes over both sides during iterative unmasking---the frozen prefix and suffix tokens provide bidirectional conditioning at every step, yielding coherent boundary transitions without any explicit stitching.

\begin{figure}[t]
    \centering
    \includegraphics[width=0.95\textwidth]{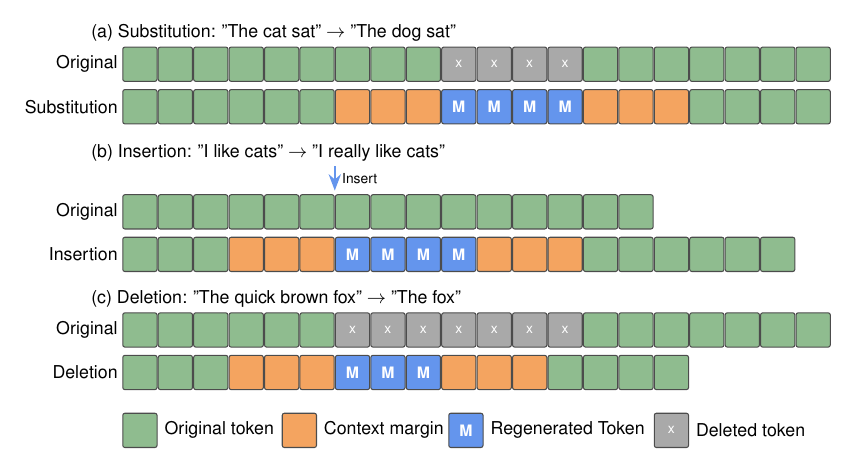}
    \caption{Speech editing pipeline: (1) align text$\to$speech via attention, (2) mask affected region with context margins, (3) regenerate via iterative unmasking. The surrounding tokens remain frozen, providing bidirectional conditioning throughout.}
    \label{fig:speech_edit}
\end{figure}

\section{Experiments}
\label{sec:experiments}

\subsection{Setup}

\paragraph{Benchmark.} We evaluate on Seed-TTS-Eval~\citep{seedtts}: \textbf{test-zh} (2{,}020 Chinese utterances, CER), \textbf{test-zh-hard} (400 adversarial Chinese, CER), and \textbf{test-en} (1{,}088 English, WER).

\paragraph{Metrics.} Character/word error rate (CER/WER) via Whisper-large-v3~\citep{whisper} (English) and Paraformer~\citep{paraformer} (Chinese); speaker similarity (SS) via WavLM-large~\citep{wavlm} cosine similarity.

\paragraph{Baselines.} We compare against representative open-source systems on the Seed-TTS-Eval leaderboard. All baseline numbers are from the leaderboard unless otherwise noted.

\subsection{Main Results}

Table~\ref{tab:main_results} presents the comparison on Seed-TTS-Eval.

\begin{table}[t]
\centering
\caption{Comparison on Seed-TTS-Eval benchmark (open-source leaderboard). SS: speaker similarity (\%, WavLM-based). $\downarrow$: lower is better; $\uparrow$: higher is better.}
\label{tab:main_results}
\small
\begin{tabular}{@{}llcccccc@{}}
\toprule
& & \multicolumn{2}{c}{\textbf{test-zh}} & \multicolumn{2}{c}{\textbf{test-en}} & \multicolumn{2}{c}{\textbf{test-zh-hard}} \\
\cmidrule(lr){3-4}\cmidrule(lr){5-6}\cmidrule(lr){7-8}
\textbf{Model} & \textbf{Type} & \textbf{CER}$\downarrow$ & \textbf{SS}$\uparrow$ & \textbf{WER}$\downarrow$ & \textbf{SS}$\uparrow$ & \textbf{CER}$\downarrow$ & \textbf{SS}$\uparrow$ \\
\midrule
Human & --- & 1.26 & 75.5 & 2.14 & 73.4 & --- & --- \\
Seed-TTS & AR & 1.12 & \textbf{79.6} & 2.25 & \textbf{76.2} & 7.59 & \textbf{77.6} \\
VoxCPM~\citep{voxcpm} & AR & \textbf{0.93} & 77.2 & \textbf{1.85} & 72.9 & 8.87 & 73.0 \\
Index-TTS\,2~\citep{indextts2} & AR & 1.03 & 76.5 & 2.23 & 70.6 & 7.12 & 75.5 \\
MaskGCT & NAR & 2.27 & 77.4 & 2.62 & 71.4 & 10.27 & 74.8 \\
F5-TTS & Flow & 1.52 & 74.1 & 2.00 & 64.7 & 8.67 & 71.3 \\
CosyVoice\,3-0.5B & AR & 1.21 & 78.0 & 2.24 & 71.8 & \textbf{6.71} & 75.8 \\
\midrule
LLaDA-TTS (64 steps) & Diffusion & 0.98 & 74.6 & 1.96 & 71.1 & 7.04 & 73.5 \\
\bottomrule
\end{tabular}
\end{table}

LLaDA-TTS achieves 0.98\% CER on test-zh, outperforming the CosyVoice\,3-0.5B base model (1.21\%) and approaching the strongest AR systems on test-en (1.96\% vs.\ 1.85\% for VoxCPM). It substantially outperforms MaskGCT (the only other NAR baseline) on all metrics. On the adversarial test-zh-hard set, LLaDA-TTS (7.04\%) closely matches the CosyVoice\,3 baseline (6.71\%), demonstrating that masked diffusion can achieve comparable robustness to AR decoding even on challenging inputs such as tongue twisters. Overall, these results confirm that replacing the AR decoder with masked diffusion preserves speech quality across all evaluation conditions.

\subsection{Speed--Quality Tradeoff}
\label{sec:exp_speed}

A key advantage of LLaDA-TTS is controlling quality--speed by adjusting the step count $T$. Figure~\ref{fig:speed_quality} shows both the CER and RTF as a function of denoising steps.

\begin{figure}[t]
    \centering
    \includegraphics[width=0.85\textwidth]{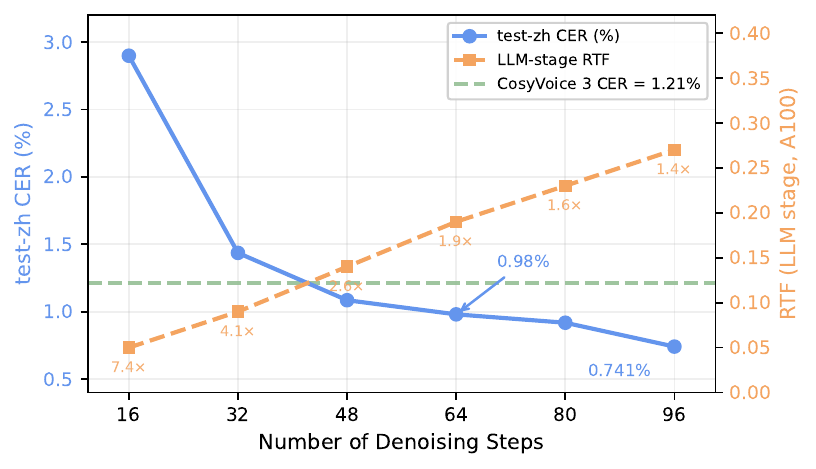}
    \caption{Speed--quality tradeoff. Left axis: test-zh CER (\%) vs.\ denoising steps; right axis: LLM-stage RTF on A100. The dashed green line marks the CosyVoice\,3 AR baseline CER (1.21\%). LLaDA-TTS surpasses the AR baseline at 48 steps (${\sim}2.6{\times}$ speedup) and achieves 0.74\% CER at 96 steps.}
    \label{fig:speed_quality}
\end{figure}

LLaDA-TTS surpasses the AR baseline CER at 48 steps (1.09\% vs.\ 1.21\%, ${\sim}2.6{\times}$ LLM-stage speedup). At 64 steps, it achieves 0.98\% CER with a $1.9{\times}$ speedup; at 96 steps, CER drops further to 0.74\%. Crucially, inference cost is independent of output length: for a 250-token utterance, the AR model requires 250 sequential passes while LLaDA-TTS uses exactly $T$.

\subsection{Ablation: Initialization and Label Shift}
\label{sec:exp_ablation}

Table~\ref{tab:ablation} isolates the contributions of AR initialization and label shift.

\begin{table}[t]
\centering
\caption{Ablation of initialization strategy and label shift on Seed-TTS-Eval mini set at matched training compute.}
\label{tab:ablation}
\small
\begin{tabular}{lccc}
\toprule
\textbf{Configuration} & \textbf{zh CER}$\downarrow$ & \textbf{zh-hard CER}$\downarrow$ & \textbf{en WER}$\downarrow$ \\
\midrule
AR weights only (epoch 0, no training) & 99.97 & 99.94 & 99.97 \\
From scratch & 45.27 & 68.25 & 62.58 \\
AR init, no label shift & 1.45 & 12.33 & 2.30 \\
\midrule
AR init + label shift (full) & 0.98 & 7.04 & 1.96 \\
\bottomrule
\end{tabular}
\end{table}

AR initialization with label shift provides the strongest results, dramatically outperforming training from scratch (0.98\% vs.\ 45.27\% CER). Removing label shift noticeably degrades performance (1.45\% CER, 12.33\% zh-hard), confirming that the shifted prediction convention is the critical bridge enabling knowledge transfer from AR pretraining.

These results empirically validate Theorem~\ref{thm:kl_bound}: the $\varepsilon$-forward dependence of acoustic tokens means that the AR model's left-to-right predictions are already close to optimal for the bidirectional masked prediction task, enabling rapid convergence when combined with the masked diffusion objective. Label shift further aligns the weight geometry so that the pretrained output projection can be efficiently adapted to the new bidirectional regime. Without label shift, the output projection must re-learn the prediction mapping from misaligned hidden representations, explaining the performance gap.

\subsection{Emergent AR-like Behavior and Alignment}
\label{sec:exp_emergent}

\paragraph{AR-like unmasking order.} Despite fully bidirectional attention, the unmasking process exhibits a predominantly left-to-right progression (Figure~\ref{fig:unmasking}). Tokens near the utterance start are resolved first, with a ``wavefront'' sweeping rightward---closely resembling AR decoding. However, the ordering is not strictly sequential: high-confidence positions (prosodic boundaries, silence) are resolved earlier than their position predicts, while ambiguous tokens are deferred. This indicates that LLaDA-TTS inherits sequential priors from AR pretraining but modulates them with confidence-based lookahead, complementing sequential generation with global bidirectional context.

Theorem~\ref{thm:kl_bound} provides a principled explanation: the $\varepsilon$-forward dependence of acoustic tokens ensures that left-context-only predictions are near-optimal. During early unmasking stages, positions with small indices benefit from the always-visible text and prompt prefix, yielding low KL divergence and high confidence---hence early unmasking. As these early tokens become revealed, they provide additional left context for subsequent positions, creating a cascading wavefront effect (Corollary~\ref{cor:wavefront}). The deviations from strict left-to-right order arise precisely when bidirectional context does improve prediction---e.g., at prosodic boundaries where surrounding revealed tokens provide strong bilateral constraints.

\begin{figure}[t]
    \centering
    \includegraphics[width=0.95\textwidth]{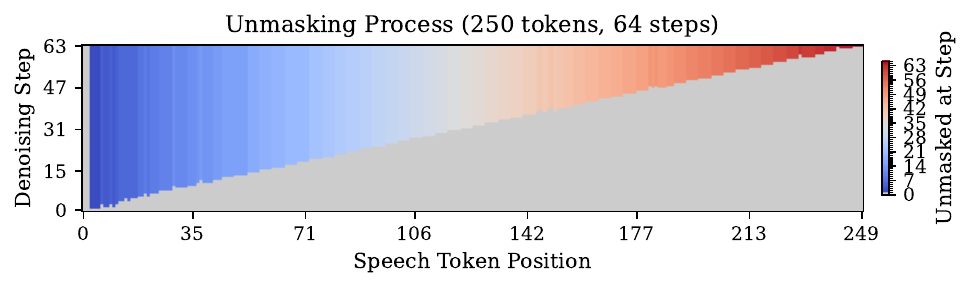}
    \caption{Unmasking process for a Chinese utterance (64 steps). Each column is a speech token position; color indicates the step at which the token was unmasked (blue=early, red=late, gray=still masked). The generation exhibits a predominantly left-to-right progression with confidence-based deviations.}
    \label{fig:unmasking}
\end{figure}

\paragraph{Attention-based alignment.} We discover that, after masked diffusion fine-tuning, certain attention heads in LLaDA-TTS spontaneously develop precise text-to-speech alignment. Figure~\ref{fig:alignment_viz} (left) visualizes the head-averaged attention map for six evenly spaced layers: most layers show diffuse attention, but layer~11 produces a sharp diagonal pattern indicating monotonic text-to-speech correspondence. Figure~\ref{fig:alignment_viz} (right) further dissects layer~11 by showing each of the first 6 individual heads: only H1 and H5 exhibit clean diagonal alignment, while other heads show no such structure. Evaluated against MMS forced-alignment ground truth~\citep{mms}, attention-based alignment achieves MAE\,=\,1.29 tokens ($\approx$52\,ms at 25\,Hz), substantially outperforming proportional alignment (MAE\,=\,2.99).

\begin{figure}[t]
    \centering
    \begin{subfigure}[t]{0.48\textwidth}
        \centering
        \includegraphics[width=\textwidth]{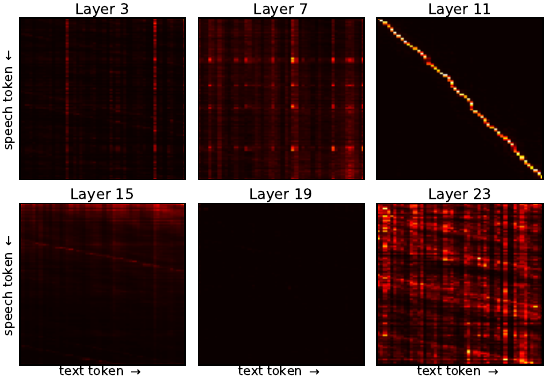}
        \caption{Head-averaged attention for selected layers. Layer~11 shows sharp diagonal alignment.}
    \end{subfigure}
    \hfill
    \begin{subfigure}[t]{0.48\textwidth}
        \centering
        \includegraphics[width=\textwidth]{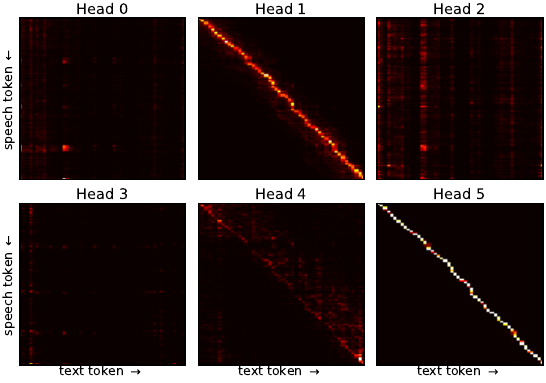}
        \caption{Individual heads (H0--H5) of layer~11. H1 and H5 develop monotonic alignment.}
    \end{subfigure}
    \caption{Emergent text-to-speech alignment in LLaDA-TTS. Axes: text tokens ($\rightarrow$) vs.\ speech tokens ($\downarrow$).}
    \label{fig:alignment_viz}
\end{figure}

\subsection{Speech Editing}
\label{sec:exp_editing}

We evaluate LLaDA-TTS's zero-shot speech editing on the Ming-Freeform-Audio-Edit-Benchmark~\citep{mingbenchmark}, which covers three editing operations: insertion, deletion, and substitution. We measure word error rate (WER) on the full edited utterance via Whisper-large-v3, speaker similarity (SIM) via WavLM cosine similarity between the edited and original speech, and DNSMOS for perceptual quality. We compare against VoiceCraft~\citep{voicecraft} (AR, requires specialized editing training) and Ming-UniAudio~\citep{minguniaudio} (Audio-LLM).

\section{Conclusion}
\label{sec:conclusion}

We presented LLaDA-TTS, which replaces the autoregressive LLM decoder in LLM-based TTS with a masked diffusion model. By fine-tuning from AR weights with bidirectional attention and a $1/t$-weighted masked prediction objective, LLaDA-TTS achieves competitive speech quality at 64 steps while decoupling cost from output length. The same architectural change that enables acceleration also enables native speech editing---speed and editability are two manifestations of the same bidirectional masked design.

Two emergent findings deepen our understanding: (1)~the diffusion model develops AR-like generation dynamics, retaining sequential priors from pretraining while adding confidence-based flexibility; (2)~precise attention alignment emerges at specific layer-head pairs without supervision. Our theoretical analysis (Theorem~\ref{thm:kl_bound}) provides a principled explanation: the $\varepsilon$-forward dependence of acoustic tokens ensures that AR predictions are provably close to the optimal bidirectional posterior, explaining both the effectiveness of initialization and the emergent left-to-right dynamics. Together, these findings suggest that the transition from AR to masked diffusion preserves and augments the structural knowledge learned during AR pretraining.

\paragraph{Generalizability.} Our method modifies only the attention mask and loss---leaving tokenizer, prompt format, and acoustic model unchanged. Any LLM-based AR TTS system using discrete tokens (VALL-E, Seed-TTS, Spark-TTS, etc.) could be similarly converted, inheriting the benefits of step-independent cost and native editing.

\paragraph{Limitations.} (1) LLaDA-TTS requires specifying output length in advance (via a token-text ratio). (2) The current implementation is non-streaming.

\paragraph{Future work.} Promising directions include targeted post-training for difficult test sets, step distillation below 32 steps, streaming via prefix-priority unmasking, and scaling to larger backbones and datasets.

\bibliographystyle{plainnat}
\bibliography{references}

\end{document}